\documentclass[11pt,reqno]{amsart}
\usepackage{amsmath,amsthm,amssymb,bm,bbm,dsfont,braket}
\usepackage[mathscr]{eucal}
\usepackage{amsaddr}
\usepackage{upgreek}
\usepackage[left=2cm,right=2cm,top=2cm,bottom=2cm]{geometry}
\usepackage{mathtools}\mathtoolsset{centercolon}
\mathtoolsset{showonlyrefs}

\usepackage{cite}

\usepackage[shortlabels]{enumitem}
\usepackage{setspace}
\usepackage{graphicx}
\usepackage{verbatim}
\usepackage{float}
\usepackage{placeins}
\usepackage{array}
\usepackage{booktabs}
\usepackage{threeparttable}
\usepackage[update,prepend]{epstopdf}
\usepackage{multirow}
\usepackage{amsfonts,amssymb,dsfont,xfrac,comment}
\usepackage[abs]{overpic}

\usepackage[usenames,dvipsnames]{color}
\usepackage[hidelinks]{hyperref}
\hypersetup{
	unicode=false,          
	pdftoolbar=true,        
	pdfmenubar=true,        
	pdffitwindow=false,     
	pdfstartview={FitH},    
	pdftitle={My title},    
	pdfauthor={Author},     
	pdfsubject={Subject},   
	pdfcreator={Creator},   
	pdfproducer={Producer}, 
	pdfkeywords={keyword1} {key2} {key3}, 
	pdfnewwindow=true,      
	colorlinks=true,        
	linkcolor=Red,          
	citecolor=ForestGreen,  
	filecolor=Magenta,      
	urlcolor=BlueViolet,    
}
\usepackage{doi}
\usepackage{url}
\usepackage{caption, subcaption}
\usepackage{enumitem}
\usepackage{shadethm}

\makeatletter
\ifx\@NODS\undefined%

\let\mathbb=\mathds
\else%
\fi
\makeatother

\allowdisplaybreaks

\DeclareMathOperator{\Tr}{Tr}

\DeclareMathOperator{\e}{\mathrm{e}}

\newcommand{\be}{{\mathbf e}}

\newcommand{\tr}{\operatorname{Tr}}

\newcommand{\ten}{\otimes}
\newcommand{\pl}{\hspace{.1cm}}

\def\0{{\mathbf{0}}}
\def\1{{\mathbf{1}}}
\def\2{{\mathbf{2}}}
\def\3{{\mathbf{3}}}
\def\4{{\mathbf{4}}}
\def\5{{\mathbf{5}}}
\def\6{{\mathbf{6}}}

\def\7{{\mathbf{7}}}
\def\8{{\mathbf{8}}}
\def\9{{\mathbf{9}}}

\def\be{\begin{equation}}
\def\ee{\end{equation}}
\def\bea{\begin{eqnarray}}
\def\eea{\end{eqnarray}}

\def\eps{\varepsilon}

\newcommand{\id}{\operatorname{id}}

\theoremstyle{plain}
\newtheorem{theo}{Theorem} 
\newtheorem{prop}[theo]{Proposition} 
\newtheorem{lemm}[theo]{Lemma} 
\newshadetheorem{theorem}{Theorem}

\theoremstyle{definition}
\newtheorem{defn}[theo]{Definition} 

\theoremstyle{remark}
\newtheorem{remark}{Remark}[section]

\numberwithin{equation}{section}

\makeatletter
\newcommand{\opnorm}{\@ifstar\@opnorms\@opnorm}
\newcommand{\@opnorms}[1]{%
	$\left|\mkern-1.5mu\left|\mkern-1.5mu\left|
	#1
	\right|\mkern-1.5mu\right|\mkern-1.5mu\right|$
}
\newcommand{\@opnorm}[2][]{%
	\mathopen{#1|\mkern-1.5mu#1|\mkern-1.5mu#1|}
	#2
	\mathclose{#1|\mkern-1.5mu#1|\mkern-1.5mu#1|}
}
\makeatother

\begin{document}

\let\origmaketitle\maketitle
\def\maketitle{
	\begingroup
	\def\uppercasenonmath##1{} 
	\let\MakeUppercase\relax 
	\origmaketitle
	\endgroup
}

\title{\bfseries \Large{ 
Optimal Second-Order Rates for  Quantum Information Decoupling
		}}

\author{ \normalsize \textsc{Yu-Chen Shen$^{1}$, Li Gao$^{6}$, and Hao-Chung Cheng$^{1\textrm{--}5}$}}
\address{\small  	
$^{1}$Department of Electrical Engineering and Graduate Institute of Communication Engineering,\\ National Taiwan University, Taipei 106, Taiwan (R.O.C.)\\
$^{2}$Department of Mathematics, National Taiwan University\\
$^{3}$Center for Quantum Science and Engineering,  National Taiwan University\\
$^4$Physics Division, National Center for Theoretical Sciences, Taipei 10617, Taiwan (R.O.C.)\\
$^{5}$Hon Hai (Foxconn) Quantum Computing Center, New Taipei City 236, Taiwan\\
$^6$School of Mathematics and Statistics, Wuhan University, Hubei Province 430072, P.~R.~China\\
}


\email{\href{mailto:ycshen089@gmail.com}{ycshen089@gmail.com}}
\email{\href{mailto:gaolimath@gmail.com}{gaolimath@gmail.com}}
\email{\href{mailto:haochung.ch@gmail.com}{haochung.ch@gmail.com}}

\date{\today}
\begin{abstract}
In this paper, we consider the standard quantum information decoupling, in which Alice aims to decouple her system from the environment by local operations and discarding some of her systems. 
To achieve an \(\varepsilon\)-decoupling with trace distance as the error criterion, we establish a near-optimal one-shot characterization for the largest dimension of the remainder system in terms of the conditional \((1-\varepsilon)\)-hypothesis-testing entropy.
When the underlying system is independent and identically prepared, our result leads to the matched second-order rate as well as the matched moderate deviation rate.
As an application, we find an achievability bound in entanglement distillation protocol, where the objective is for Alice and Bob to transform their quantum state to maximally entangled state with largest possible dimension using only local operations and one-way classical communications.

\end{abstract}
\maketitle

\section{Introduction} \label{sec:introduction}
Given a bipartite quantum system $A$ and $E$, quantum decoupling aims to un-correlate the system $A$ from the other system $E$ via suitable
processes operated on $A$. It has considerable applications in quantum information theory and quantum physics.
For information theory, it has been utilized in quantum state merging \cite{HOW07}, quantum reverse Shannon theorem \cite{BCR11, BDH+09}, and characterizing correlation between systems \cite{Bus09, GPW05}. As for quantum physics, quantum decoupling has been related to the process to thermal equilibrium in thermal dynamics \cite{LPSW09}, black hole radiation \cite{HP07}, and certain topics in solid state physics \cite{BH13}.

In this paper, we consider a standard quantum decoupling protocol: for a bipartite quantum state $\rho_{AE}$ held by Alice and Eve, Alice applies a unitary map ${U}_A$ randomly chosen from the unitary group $\mathds{U}({A})$ on Hilbert space ${A}$, followed by a partial trace $\mathcal{T}_{A\to C} = \Tr_{A_1}$ (suppose Alice's system $A=A_1 C$ consists of two subsystems $A_1$ and $C$), to decouple her system from Eve's. To ensure that the processed state is almost indistinguishable from the ideal decoupled state, we adopt \emph{trace distance} as an error criterion to measure the closeness of these two states. That is, for any bipartite state $\rho_{AE}$, we define the average decoupling error
    \begin{align}
        \Delta(A|E)_{\rho}
        &:=\frac{1}{2}\mathds{E}_U\left\|\mathcal{T}_{A\to C}(U_A\rho_{AE}U_A^{\dagger}) - \omega_{C}\otimes \rho_E\right\|_1 \notag\\
        &= \frac{1}{2}\int_{\mathds{U}({A})} \left\|\mathcal{T}_{A\to C}(U_A(\rho_{AE}))-\omega_{C}\otimes\rho_E\right\|_1 \,\mathrm{d} U.\notag
    \end{align}
where $\mathrm{d} U$ is respect to the Haar measure, and $\omega_{A'C}$ denotes the Choi state of $\mathcal{T}_{A\to C}$ ($\omega_C = \frac{\mathds{1}_C}{|{C}|}$ in our case). 
Note that the Haar integral can be replaced by any unitary $2$-design. We are concerned about the \emph{maximal remainder dimension} $|{C}|$ of the output Hilbert space ${C}$ for an  $\varepsilon$-decoupling \cite[\S 7]{Tom16}:
\begin{align} 
    \ell^{\varepsilon}(A{\,|\,}E)_{\rho} := \sup\left\{   \left|{C}\right| \in \mathds{N} : \Delta(A{\,|\,}E)_{\rho} \leq \varepsilon \right\}.\label{def:maxlq}
\end{align}

Due to the importance of quantum decoupling, there have been extensive studies on the \emph{one-shot characterizations} for $\ell^{\varepsilon}(A{\,|\,}E)_{\rho}$.
Tight one-shot and asymptotic characterizations have been obtained when \emph{purified distance} \cite{TCR10} is used as the error criterion \cite{TH13, MBD+17, KL21, LY21a}. 
As for trace distance in the same manner with this work, Dupuis \textit{et al.}~\cite{Dup14, MBD+17} proved a one-shot achievability (i.e.~lower) bound to $\ell^{\varepsilon}(A{\,|\,}E)_{\rho}$ and also a converse (i.e.~upper) bound was in \cite[Theorem~4.1]{Dup14}. 
 Nevertheless,  those bounds {disagree} in the second-order term. This suggests that the existing one-shot bounds are not tight enough in the asymptotic limit.
Hence, {\emph{tight}}  one-shot characterizations of $\ell^{\varepsilon}$ in trace distance have been open for a decade.

We resolve this open problem of tight one-shot characterizations of decoupling in trace distance error. Our main result is a lower and upper bound of one-shot characterization of maximal remainder dimension:
\begin{align}
    \log \ell^{\varepsilon}(A{\,|\,}E)_{\rho} \approx \frac{1}{2}\left(\log|{A}|+H_\textnormal{h}^{1-\varepsilon \pm \delta}(A{\,|\,}E)_{\rho}  \right),\notag
\end{align}
where ``$\approx$'' means equality up to some logarithmic additive terms, and $\delta$ is a parameter that can be chosen for optimization, and $H_{\mathrm{h}}^{\varepsilon}$ is $\varepsilon$-\emph{hypothesis testing entropy} \cite{TH13} defined later in \eqref{eq:one-shot_entropy}. Moreover, in i.i.d. scenario, our result leads to optimal second-order asymptotic rate. Namely,
\begin{align}
\log \ell^{\varepsilon}(A^n|E^n)_{\rho^{\otimes n}} &= n\left( \frac{1}{2}(\log |{A}|+H(A|E)_{\rho})\right) + \frac{1}{2}\sqrt{nV(A|E)_{\rho}}\, \Phi^{-1}(\varepsilon) + O(\log n),\notag
\end{align}
where $H(A|E)_{\rho}$ is the conditional entropy and $V(A|E)_{\rho}$ is the conditional information variance defined later in \eqref{eq:conditional entropy}, and $\Phi^{-1}$ is the inverse of cumulative normal distribution \eqref{eq:inverse_Phi}.
Moreover, our one-shot characterization applies to moderate deviation regime, where we derive the optimal rate of the maximal remainder dimension $\ell^{\varepsilon}$ when $\varepsilon$ approaches zero moderately fast as $n\to \infty$.

As an application, our result of quantum decoupling leads to an achievability bound of \emph{entanglement distillation}. Suppose Alice and Bob hold a quantum state $\rho_{AB}$. The goal of entanglement distillation is for Alice and Bob to transform $\rho_{AB}$ into a maximally entangled state by performing a set of operations. In our work we consider Alice and Bob only apply local operations and one-way classical communications denoted as $\Pi_{1-\text{LOCC}}$. We consider the accuracy criterion as the fidelity between the resulted state and the maximally entangled state
\begin{align}
F_{\textnormal{ED}}(\rho,m) := \sup_{\Pi_{1-\text{LOCC}}} F\left(\Phi^{m},\Pi_{1-\text{LOCC}}(\rho_{AB})\right),
\end{align}
where $\Phi^m$ denotes the maximally entangled state with dimension $m$. Similar to quantum decoupling, in entanglement distillation we are concerned about the $\varepsilon$-error maximal distillable entanglement, defined as follows.
\begin{align}
E_{\textnormal{ED}}^{\varepsilon}(\rho) := \sup\left\{M \in \mathds{N}: m\geq M \wedge F_{\textnormal{ED}}(\rho,m) \geq 1-\varepsilon\right\}.\notag
\end{align}

Plenty of studies have been working on entanglement distillation with different sets of operations. Various approaches \cite{R99, R01, HHH00, CW04, LDS18, WD16} have been developed to evaluate the maximal distillable entanglement. Moreover, there are also papers working on bounding asymptotic expansion of maximal distillable entanglement \cite{FWTD19, DL15}. 
In this paper, we give an lower bound of maximal distillable entanglement $E_{\textnormal{ED}}^{\varepsilon}$.
Applying the result of our maximal remainder dimension, the following holds for one-shot maximal distillable entanglement,
\begin{align}
			\log E_{\textnormal{ED}}^{\varepsilon}(\rho) \gtrsim H_\textnormal{h}^{1-\varepsilon - \delta}(A{\,|\,}E)_{\rho}.\notag
\end{align}
where $\rho_{ABE}$ is the purified state of $\rho_{AB}$.
Here ``$\gtrsim$'' means inequality up to some logarithmic additive terms.
This lower bound leads to the following second-order asymptotic bound in i.i.d. scenario,
\begin{align}
\log E_{\textnormal{ED}}^{\varepsilon}(\rho^{\otimes n}) \geq nH(A|E)_{\rho} + \sqrt{nV(A|E)_{\rho}}\Phi^{-1}(\varepsilon)  + O(\log n).\notag
\end{align}

The paper is structured as follows. In the rest of this section we compare our results with existing literature. Section \ref{sec:notation} reviews the necessary background on entropic quantities and relations between divergence. In Section~\ref{sec:sc},  we present our main result: a near optimal one-shot characterization of maximal remainder dimension, as well as the asymptotic bound which our one-shot bound leads to. In Section~\ref{sec:lower} and Section~\ref{sec:converse} we show the proof of achievability bound and the converse bound of our one-shot characterization, respectively. In Section~\ref{sec:distill} we show our application of quantum decoupling to entanglement distillation. We conclude our paper in Section~\ref{sec:conclusion}.

\subsection{Comparison with existing results}
For quantum decoupling, in \cite[Theorem~3.1]{Dup14}, Dupuis et al. found a one-shot direct bound as follows:
\begin{align}
\Delta(A|E)_{\rho}
\leq \frac{1}{2}\e^{-\frac{1}{2}H^{\varepsilon}_{\min}(A|E)_{\rho}+\frac{1}{2}\log\frac{|{C}|^2}{|{A}|}}+3\varepsilon,\notag
\end{align}
where $H^{\varepsilon}_{\min}(A|E)_{\rho}$ is the \emph{smooth conditional min-entropy} of $A$ given $E$.
Such result leads to the lower bound of $\ell^{\varepsilon}(A|E)_{\rho}$ \cite{CHR18}:
\begin{align}
\log\ell^{\varepsilon}(A|E)_{\rho} \gtrsim \frac{1}{2}(\log|{A}| + H^{\frac{\varepsilon}{3} }_{\min}(A|E)_{\rho}).\notag
\end{align}
By second-order expansion of $H^{\varepsilon}_{\min}$ \cite{TH13}, this one-shot bound implies to the following second-order asymptotic bound,
\begin{align}
\log \ell^{\varepsilon}(A^n|E^n)
&\geq n\left( \frac{1}{2}(\log |{A}|+H(A|E)_{\rho})\right) + \frac{1}{2}\sqrt{nV(A|E)_{\rho}}\Phi^{-1}({\varepsilon^2}/{9}) + O(\log n).\notag
\end{align}
On the other hand, applying \cite[Corollary~4.2]{Dup14} with $\varepsilon''$ being close to $0$, one have
\begin{align}
\log\ell^{\varepsilon}(A|E)_{\rho}\leq \frac{1}{2}\left(H^{2\sqrt{2\varepsilon}+2\sqrt{\varepsilon'}}_{\min}(A|E)_{\rho}+\log|{A}|\right) + O(\log\varepsilon').\notag
\end{align}
Similarly, above bound implies the following second-order asymptotic bound
\begin{align}
\log \ell^{\varepsilon}(A^n|E^n)\leq n\left( \frac{1}{2}(\log |{A}|+H(A|E)_{\rho})\right) + \frac{1}{2}\sqrt{nV(A|E)_{\rho}}\Phi^{-1}(8\varepsilon) + O(\log n).\notag
\end{align}

Note that the second-order term of the upper bound and lower bound do not match. This is because the smoothing parameters $\frac{\eps}{3}$ and $\eps''=2\sqrt{2\varepsilon}+2\sqrt{\varepsilon'}$ are quantitatively different.
In some circumstances, these \textit{unmatched} parameters cause significant gaps between the upper and lower bounds \textit{even in the one-shot setting}, hence never achieve matched second-order rate in the i.i.d.~asymptotic limit. For the same reason, a tight characterization in purified distance \cite{TH13, MBD+17, KL21, LY21a} does not help here because converting purified distance error criterion to trace distance using Uhlmann's theorem also leads to unmatched error parameters.

In comparison, our result (Theorem~\ref{theo:one-shot_second-order_PA}) is optimal in second-order expansion with trace distance as security standard (Proposition~\ref{qdecoupleiid}). This suggests that our one-shot bound, without going through smooth entropies, is tighter in both directions. 

\section{Notation and Information-Theoretic Quantities} \label{sec:notation}
For an integer $M\in\mathds{N}$, we denote $[M]:= \{1,\ldots, M\}$. We use `$\wedge$' to denote the minimum between two scalars or the conjunction `and' between two statements. 
	We denote $\mathds{1}_{A}$ as the indicator function for a condition $A$.
	The density operators are positive semi-definite operators with unit trace. For Hermitian operators $P$ and $Q$, we denote $\{P\leq Q\}$ as the projection operator on the expansion of eigenvectors with non-negative eigenvalues of the matrix $Q-P$.
	For a trace-class operator $H$, the trace class norm (also called Schatten-$1$ norm) is defined by
		$\|H\|_1 := \Tr\left[ \sqrt{ H^\dagger H } \right]$. For density operators $\rho$ and $\sigma$, the fidelity is defined as $F(\rho,\sigma):= \|\sqrt{\rho}\sqrt{\sigma}\|_1$.
	A quantum operation is mathematically described as a completely positive and trace preserving (CPTP) map on density operators.

     For two density operators $\rho$ and $\sigma$, we define the $\varepsilon$-\emph{information spectrum divergence} \cite{HN03, NH07} as
	\begin{align}
		D_\text{s}^{\varepsilon}(\rho \parallel \sigma) &:= \sup_{c\in\mathds{R}} \left\{ \log c : \Tr\left[ \rho \left\{ \rho \leq c \sigma \right\} \right] \leq \varepsilon  \right\}, \label{eq:Ds}
	\end{align}
	and the $\varepsilon$-\emph{hypothesis testing divergence} \cite{TH13, WR13, Li14} as
	\begin{align} 
		D_\text{h}^\varepsilon(\rho \,\|\, \sigma) := \sup_{0\leq T\leq \mathds{1} } \left\{ -\log \Tr[\sigma T] : \Tr[\rho T] \geq 1 - \varepsilon \right\}.\notag
	\end{align}

 The \emph{quantum relative entropy} \cite{Ume62, HP91} and \emph{quantum relative entropy variance} \cite{TH13,Li14} for density operator $\rho$ and positive definite operator $\sigma$ are defined as
	\begin{align}
		D(\rho \,\|\,\sigma) &:= \Tr\left[ \rho \left( \log \rho - \log \sigma \right) \right];\notag\\
		V(\rho \,\|\,\sigma) &:= \Tr\left[ \rho \left( \log \rho - \log \sigma \right)^2 \right] - \left( D(\rho \,\|\,\sigma) \right)^2.\notag
	\end{align}

For a bipartite quantum system $AB$, we recall the \emph{conditional $\varepsilon$-hypothesis testing entropy} \cite{TH13} as
	\begin{align}  
		H_\text{h}^\varepsilon(A{\,|\,}B)_\rho := -D_\text{h}^\varepsilon\left(\rho_{AB} \,\|\,\mathds{1}_A\otimes \rho_B\right).\label{eq:one-shot_entropy}
	\end{align}
	We note that $H_\text{h}^\varepsilon(A{\,|\,}B)_\rho$ can be formulated as a semi-definite optimization problem \cite{TH13}, \cite[\S 1]{Wat18}.
	The \emph{quantum conditional entropy} and the \emph{quantum conditional information variance} of $\rho_{AB}$ are defined, respectively as
	\begin{align} \label{eq:conditional entropy}
		\begin{split}
		H(A{\,|\,} B)_\rho &:= -D(\rho_{AB} \,\|\,\mathds{1}_A\otimes \rho_B)\,; \\
		V(A{\,|\,} B)_\rho &:= V(\rho_{AB} \,\|\,\mathds{1}_A\otimes \rho_B)\,.
		\end{split}
	\end{align}

For any self-adjoint operator $H$ with eigenvalue decomposition $H = \sum_i \lambda_i |e_i\rangle \langle e_i|$,
	we define the set $\textnormal{\texttt{spec}}(H):= \{\lambda_i\}_i$ to be the set of eigenvalues of $H$, and $|\textnormal{\texttt{spec}}(H)|$ to be the number of distinct eigenvalues of $H$.
	We define the \emph{pinching map} with respect to $H$ as
	\begin{align}
		\mathscr{P}_H[L] : L\mapsto \sum_{i=1} e_i L e_i\, \label{eq:pinching}
	\end{align}
	where $e_i$ is the spectrum projection onto $i$th distinct eigenvalues.

For positive semi-definite operator $K$ and positive operator $L$, we use the following short notation for the \emph{noncommutative quotient}.
\begin{align} 
    \frac{K}{L}:= L^{-\frac{1}{2}} K L^{-\frac{1}{2}}\,.\notag
\end{align} 
For positive semi-definite operator $\rho$ and  positive definite $\sigma$, we define the \emph{collision divergence}\footnote{Note that for unnormalized $\rho$, the collision divergence is usually defined as $\log \Tr[ ( \sigma^{-\frac{1}{4}} \rho  \sigma^{-\frac{1}{4}} )^2 ] - \log \Tr[\rho]$. However, we do not use such a definition for notational convenience. 
		Indeed, our derivations will rely on the joint convexity of $\exp D_2^*$ (see Lemma~\ref{lemm:joint_convexity}) which holds for positive semi-definite operator $\rho$ as well.} \cite{Ren05} as
	\begin{align} 
		D^*_2(\rho \,\|\, \sigma) :=  \log \Tr\left[ \left( \sigma^{-\frac{1}{4}} \rho  \sigma^{-\frac{1}{4}} \right)^2 \right].\notag
	\end{align}
For any positive operator $\rho_{AB}$, the \emph{conditional min-entropy} of $A$ given $B$ is defined as
\begin{align}
    H_{\min}(A|B)_{\rho} := -\log \min \left\{\Tr[\sigma_{B}]: \sigma_B \geq 0\ ,\  \rho_{AB} \leq \mathds{1}_A\otimes \sigma_B \right\}.\notag
\end{align}
	
	\begin{lemm}[A Cauchy--Schwartz type inequality {\cite[Exercise 6.9]{Hay17}}] \label{lemm:Cauchy--Swartz}
		For any operator $X$ and a positive definite density operator $\rho$, it holds that
		\begin{align}
			\left\|X\right\|_1 \leq \sqrt{ \Tr\left[X^\dagger X \rho^{-1} \right] }.\notag
		\end{align}
	\end{lemm}

\begin{lemm}[{Pinching inequality} {\cite{Hay02}}] \label{lemm:pinching}
		For every $d$-dimensional self-adjoint operator $H$ and positive semi-definite operator $L$,
		\begin{align} 
			\mathscr{P}_H[L] \geq \frac{1}{|\textnormal{\texttt{spec}}(H)|} L\,.\notag
		\end{align}
		Moreover, it holds that for every $n\in\mathds{N}$, 
		\begin{align}
			\left|\textnormal{\texttt{spec}}\left(H^{\otimes n}\right)\right| \leq (n+1)^{d-1}\,.\notag
		\end{align}
	\end{lemm}
	
	We have the following relation between divergences that will be used in our proofs.
	\begin{lemm}[Relation between divergences {\cite[Lemma~12, Proposition~13, Theorem~14]{TH13}}] \label{lemm:relation}
		For every density operator $\rho$, positive semi-definite operator $\sigma$, $0< \varepsilon < 1$, and $0 < \delta < 1-\varepsilon$, we have
		\begin{align}
			&D_\textnormal{s}^{\varepsilon - \delta}\left( \mathscr{P}_{\sigma}[\rho] \,\|\,\sigma \right) \leq D_\textnormal{h}^{\varepsilon + \delta}(\rho \,\|\,\sigma) + \log |\textnormal{\texttt{spec}}(\sigma)| + 2 \log \delta\,; \notag\\
			&D_\textnormal{h}^\varepsilon(\rho \,\|\,\sigma) \leq D_\textnormal{s}^{\varepsilon + \delta} (\rho \,\|\,\sigma) - \log \delta\,.\notag
		\end{align}
	\end{lemm}

  \begin{lemm}[Second-order expansion \cite{TH13, Li14}]\label{lemm:second}
    For every density operator $\rho$, positive definite operator $\sigma$, $0<\varepsilon<1$, and $\delta = O(\frac{1}{\sqrt{n}})$, we have the following expansion:
    \begin{align}
        D_\textnormal{h}^{\varepsilon \pm \delta} \left(\rho^{\otimes n} \,\|\, \sigma^{\otimes n} \right) = n D\left(\rho \,\|\, \sigma\right) + \sqrt{n V\left(\rho \,\|\, \sigma\right)} \Phi^{-1} (\varepsilon) + O\left(\log n\right),\notag
    \end{align}
    where $\Phi^{-1}$ is the inverse of cumulative normal distribution
    \begin{align} \label{eq:inverse_Phi}
    	\Phi^{-1}(\varepsilon) := \sup\{u\mid \int_{-\infty}^u \frac{1}{\sqrt{2\pi}} \mathrm{e}^{-\frac12 t^2}\mathrm{d}t\leq \varepsilon\}.
    \end{align}
    \end{lemm}
\begin{lemm}[Moderate deviations {\cite[Theorem 1]{CTT2017}}] \label{lemm:moderate}
Let $(a_n)_{n\in\mathds{N}}$ be a moderate sequence satisfying 
\begin{align} \label{eq:an}
        \lim_{n\to \infty} a_n = 0 \ , \ 
        \lim_{n\to \infty} n a_n^2 = \infty \pl. 
\end{align}
and let $\varepsilon_n := \mathrm{e}^{-n a_n^2}$. For any density operator $\rho$ and positive definite operator $\sigma$, the following two asymptotic expansions hold
\begin{align}
    &\frac{1}{n} D_\textnormal{h}^{1 - \varepsilon_n } \left(\rho^{\otimes n} \,\|\, \sigma^{\otimes n} \right) = D\left(\rho \,\|\, \sigma\right) + \sqrt{2 V\left(\rho \,\|\, \sigma\right)} a_n + o\left(a_n\right) \notag\\
    &\frac{1}{n} D_\textnormal{h}^{\varepsilon_n } \left(\rho^{\otimes n} \,\|\, \sigma^{\otimes n} \right) = D\left(\rho \,\|\, \sigma\right) - \sqrt{2 V\left(\rho \,\|\, \sigma\right)} a_n + o\left(a_n\right).\notag
\end{align}
\end{lemm}

\begin{lemm}[Joint convexity {\cite[Proposition 3]{FL13}, \cite{MDS+13, WWY14}}] \label{lemm:joint_convexity}
    The map
    \begin{align}
        (\rho,\sigma) \mapsto \exp D_2^*(\rho\,\|\,\sigma)\notag
    \end{align}
    is jointly convex on all positive semi-definite operators $\rho$ and positive definite operators $\sigma$.
    \end{lemm}
    \begin{lemm}[Swap trick]\label{lemm:swap}
    For two operators $A, A' \in \mathcal{L}({A})$, the following holds
    \begin{align}
        \Tr[AA'] = \Tr[(A\otimes A')F_A],\notag
    \end{align}
    where $F_A=\sum_{i,j}\ket{i}\bra{j}\ten \ket{j}\bra{i}$ is the swap operator between system $A$ and $A'$.
    \end{lemm}
    \begin{lemm}[Variational formula of the trace distance \cite{Hel67, Hol72}, {\cite[\S 9]{NC09}}]\label{lemm:1norm}
    For density operators $\rho$ and $\sigma$,
    \begin{align}
    \frac{1}{2}\left\|\rho- \sigma\right\|_1 = \sup_{0\leq \Pi \leq \mathds{1}} \Tr\left[\Pi(\rho-\sigma)\right].\notag
    \end{align}
    \end{lemm}
    \begin{lemm}[Lower bound on the collision divergence {\cite[Theorem~3]{BEI17}}] \label{lemm:D2toDs}
    For every $0<\eta<1$ and $\lambda_1, \lambda_2>0$, density operator $\rho$, and positive semi-definite $\sigma$, we have\footnote{Ref.~\cite[Thm.~3]{BEI17} is stated for $\lambda_1= \lambda \in (0,1)$ and $\lambda_2 = 1-\lambda$. We remark that the result applies to the case $\lambda_1, \lambda_2>0$ by following the same proof.}
    \begin{align}
        \exp D_2^*\left( \rho  \,\|\, \lambda_1 \rho + \lambda_2 \sigma \right) \geq \frac{1-\eta}{\lambda_1 + \lambda_2 \cdot \mathrm{e}^{ - D_\text{s}^\eta (\rho  \,\|\, \sigma )}  }\,.\notag
    \end{align}
    \end{lemm}
\begin{lemm}[A property of Choi operator {\cite[Proposition~3.17]{Mar20}}]\label{Choi}
For a CPTP map $\mathcal{T}_{A\to B}$, define $\Gamma^{\mathcal{T}}_{A'B}$ as its Choi operator:
\begin{align}
\Gamma^{\mathcal{T}}_{A'B} = \sum_{i,j\in \mathsf{A}}\ket{i}\bra{j}_{A'}\otimes \mathcal{T}_{A\to B}(\ket{i}\bra{j}_{A}).\notag
\end{align}
Then, for any $X_{A} \in \mathcal{L}({A})$,
\begin{align}
\left(\mathcal{T}_{A\to B}\right)(X_{A}) &=\Tr_{A'}\left[(X_{A'}\otimes \mathds{1}_B)^{T_{A'}}(\Gamma^{\mathcal{T}}_{A'B})\right]  \notag\\
&= \Tr_{A'}\left[(\Gamma^{\mathcal{T}}_{A'B})^{T_{A'}}(X_{A'}\otimes \mathds{1}_B)\right].\notag
\end{align}
where $X_{AB}^{T_A}$ is the partial $A$-transpose of $X_{AB}$.
\end{lemm}

\section{Quantum Decoupling} \label{sec:sc}
The main result in this paper is the following two-sided one-shot bound of maximal remainder dimension for an $\varepsilon$-decoupling, i.e.~$\ell^{\varepsilon}(A{\,|\,}E)_{\rho}$ introduced in \eqref{def:maxlq}.

	\begin{theo}\label{theo:one-shot_second-order_PA}
		Let $\rho_{AE}$ be a density operator on ${A}{E}$. 
		Then, for every $0<\varepsilon<1$ and $0< c < \delta< \frac{\varepsilon}{3}\wedge\frac{(1-\varepsilon)}{2}$, we have 
		\begin{align}
			\frac{1}{2}\left(\log|{A}|+H_\textnormal{h}^{1-\varepsilon + 3\delta}(A{\,|\,}E)_{\rho}  \right)-\log\frac{\nu}{\delta^2}&\leq\log \ell^{\varepsilon}(A{\,|\,}E)_{\rho}\notag\\
            &\leq 
            \frac{1}{2} \left(\log|{A}| + H^{1-\varepsilon-2\delta}_h(A|E)_{\rho} + \log\frac{(1+c)(\varepsilon + c^{-1})}{\delta(\delta-c^{-1})} \right).\notag
		\end{align}
		Here, $H^{\varepsilon}_\textnormal{h}$ is the conditional $\varepsilon$-hypothesis testing entropy defined in \eqref{eq:one-shot_entropy} and
		$\nu =  |\textnormal{\texttt{spec}}(\rho_E)|$.
	\end{theo}
The achievability (lower) and converse (upper) bound in {Theorem}~\ref{theo:one-shot_second-order_PA} are proved respectively in {Section}~\ref{sec:lower} and Section~\ref{sec:converse}.

In the i.i.d.~scenario, combining the result of {Theorem}~\ref{theo:one-shot_second-order_PA} with {Lemma}~\ref{lemm:second}, we have the following second-order asymptotics of $\log\ell^{\varepsilon}(A{\,|\,}E)_{\rho}$. 

\begin{prop}[Second-order asymptotics]\label{qdecoupleiid}
For $\varepsilon\in (0,1)$ and $\mathcal{T}_{A\to C} := \Tr_{A_1}\otimes \id_{C}$, the following asymptotic holds,
\begin{align}
\log \ell^{\varepsilon}(A^n|E^n)_{\rho^{\otimes n}} = n\left( \frac{1}{2}(\log |{A}|+H(A|E)_{\rho})\right) + \frac{1}{2}\sqrt{nV(A|E)_{\rho}}\Phi^{-1}(\varepsilon) + O(\log n)\notag
\end{align}
\end{prop}
By using Lemma~\ref{lemm:moderate}, the bounds in Theorem~\ref{theo:one-shot_second-order_PA} also applies to the moderate deviation regime \cite{CTT2017, CH17}; namely, we derive the optimal rate of the maximal required output dimension when the error approaches zero moderately fast. 
\begin{prop}[Moderate deviations] \label{prop:moderate_QD}
For any $\rho_{AE}$ and any $\varepsilon_n :=  \mathrm{e}^{-na_n^2}$, where $(a_n)_{n\in \mathds{N}}$ is defined in \eqref{eq:an}, we have the asymptotic expansion,
\begin{align} 
\begin{cases}
    &\frac1n \log \ell^{\varepsilon_n} (A^n{\,|\,}E^n)_{\rho^{\otimes n}} = \frac{1}{2}(\log|{A}|+H(A{\,|\,}E)_\rho) - \sqrt{\frac{V(A{\,|\,}E)_{\rho}}{2} } \, a_n + o\left(a_n\right);\\
    &\frac1n \log \ell^{1-\varepsilon_n} (A^n{\,|\,}E^n)_{\rho^{\otimes n}} = \frac{1}{2}(\log|{A}|+H(A{\,|\,}E)_\rho) + \sqrt{\frac{V(A{\,|\,}E)_{\rho}}{2} } \, a_n + o\left(a_n\right).\notag
\end{cases}
\end{align}
\end{prop}

\section{Proof of Achievability Bound}\label{sec:lower}
In this section we prove the achievability bound of {Theorem}~\ref{theo:one-shot_second-order_PA}.
We first define the randomizing family of channels that will be used in the proof.
	\begin{defn}[Randomizing family of channels \cite{Dup21}] \label{defn:randomizing}
		A family of CPTP maps $\left\{ \mathcal{R}^h_{A\to C} : h \in\mathscr{H} \right\}$ with distribution $p$ on $\mathcal{H}$ is $\lambda$-randomizing if, for any linear bounded operator $X_{AE}$,
		\begin{align}
			\mathds{E}_{h\sim p} \left\| \left( \mathcal{R}^h - \mathcal{U} \right)(X_{AE}) \right\|_2 \leq \lambda \left\|X_{AE} \right\|_2,\notag
		\end{align}
  where a perfectly randomizing channel $\mathcal{U}_{A\to C}$ is defined as
	$
		\mathcal{U}(X_A) = \frac{\mathds{1}_C}{\left|{C}\right|} \cdot \Tr\left[ X_A\right].
	$
	\end{defn}
\begin{remark}\label{lemm:A1universal} 
It follows from \cite[Theorem 3.3]{dupuis2010decoupling} that for the partial trace $\mathcal{T}_{A\to C} = \Tr_{A_1}\otimes \id_C$, the map $\mathcal{T}_{A\to C}(U_A(\cdot)U_A^{\dagger})$ is $|{A}_1|^{-\frac12}$-randomizing. 
\end{remark}

\begin{proof}[Proof of achievability in {Theorem}~\ref{theo:one-shot_second-order_PA}]
		
		We first claim that for any $c$ in the set $  \{c > 0 | \{\mathscr{P}_{\rho_E}(\rho_{AE}) = c\mathds{1}_A\otimes \rho_E\}=0 \}$ and any $\lambda$-randomizing channel $\mathcal{R}^h_{A\to C}$,
		\begin{align}
			\frac{1}{2}\mathds{E}_{h} \left\| \mathcal{R}^h(\rho_{AE})-\mathcal{U}(\rho_{AE}) \right\|_1 \leq \tr\left[\rho_{AE}\left\{ \mathscr{P}_{\rho_E}[\rho_{AE}] > c \mathds{1}_A\otimes \rho_E \right\}\right] + \lambda\sqrt{ c|\textnormal{\texttt{spec}}(\rho_E)|\left|{C}\right| },\label{eq:pmain}
		\end{align}
   where $\mathscr{P}_{\rho_E}[\rho_{AE}]$ is defined as \eqref{eq:pinching}. Note that we must exclude the values of $c$ for which $\{\mathscr{P}_{\rho_E}(\rho_{AE}) = c\mathds{1}_A\otimes \rho_E\} \neq 0 \}$, as the presence of these values would pose difficulties in our subsequent proof when $\rho_{AE}$ is non-invertible. With this restriction, $c$ can still be chosen from all real positive numbers excluding a countable set of numbers.
		Let $\delta \in (0,\varepsilon)$ and let
		\begin{align}
			c &= \exp\left\{D_\text{s}^{1-\varepsilon+ \delta}(\mathscr{P}_{\rho_E}[\rho_{AE}]  \,\|\,\mathds{1}_A\otimes\rho_E) + \xi\right\}\notag
		\end{align}
		for some small $\xi >0$. Then, by definition of the information spectrum divergence, \eqref{eq:Ds}, we have
		\begin{align}
			\tr\left[\rho_{AE}\left\{\mathscr{P}_{\rho_E}[\rho_{AE}]> c\mathds{1}_A\otimes \rho_E\right\}\right] < \varepsilon - \delta\,.\label{eq:pmain0}
		\end{align}
		As shown in {Remark}~\ref{lemm:A1universal}, the map $\mathcal{T}(U(\cdot))$ is $\lambda = |{A}_1|^{-\frac12}= \frac{|{C}|^{\sfrac12}}{|{A}|^{\sfrac12}}$-randomizing. Then, by \eqref{eq:pmain} and \eqref{eq:pmain0},
		we know the $\varepsilon$-secret criterion is satisfied for $|{C}| = \left\lfloor \sqrt{{\frac{|{A}|\delta^{2}}{ c|\texttt{spec}(\rho_E)|}}} \right\rfloor$, i.e.~
		\begin{align}
			\frac{1}{2}\mathds{E}_{h} \left\| \mathcal{R}^h(\rho_{AE})-\mathcal{U}(\rho_{AE}) \right\|_1 \leq \varepsilon.\notag
		\end{align}
  By the definition of maximal remainder dimension \eqref{def:maxlq}, we have the following lower bound on $\log \ell^{\varepsilon}(A{\,|\,}E)_{\rho}$:
		\begin{align}
			\log \ell^{\varepsilon}(A{\,|\,}E)_{\rho}
   &\geq \frac{1}{2}\log |{A}|-\frac12\left(D_\text{s}^{1-\varepsilon+ \delta}(\mathscr{P}_{\rho_E}[\rho_{AE}]  \,\|\,\mathds{1}_A\otimes\rho_E) + \xi\right) - \frac{1}{2}\log  |\texttt{spec}(\rho_E)| + \log \delta \notag\\
			& \geq \frac{1}{2}\left(\log |{A}| + H_\text{h}^{1-\varepsilon + 3\delta}(A{\,|\,}E)_{\rho}\right) - \frac{1}{2}\xi -\log |\texttt{spec}(\rho_E)| + 2\log \delta,\notag
		\end{align}
		where we have used Lemma \ref{lemm:relation} in the last inequality. Since $\xi>0$ is arbitrary, taking $\xi\to 0$ gives our claim of lower bound in Theorem~\ref{theo:one-shot_second-order_PA}.
  
\medskip
		We shall now prove \eqref{eq:pmain}.
	We first consider the case where $\rho_{AE}$ is invertible, since the general situation of $\rho_{AE}$ can be argued by using approximation to an invertible case.
	Let \begin{align}
		&\Pi := \left\{\mathscr{P}_{\rho_E}[\rho_{AE}]\leq c\mathds{1}_A\otimes\rho_E\right\}; \notag\\
		&\Pi^\mathrm{c} := \mathds{1} - \Pi,\notag
	\end{align}
	and denote by $\mathcal{N}^h_{A\to C} := \mathcal{R}^h - \mathcal{U}$.
	The triangle inequality of $\|\cdot\|_1$ implies that
	\begin{align}
		\frac{1}{2}\mathds{E}_h\left\|\mathcal{N}^h(\rho_{AE})\right\|_1 \leq \frac{1}{2}\mathds{E}_h\left\|\mathcal{N}^h(\Pi^{\mathrm{c}}\rho_{AE}\Pi^{\mathrm{c}})\right\|_1
		+ \frac{1}{2}\mathds{E}_h\left\|\mathcal{N}^h(\Pi\rho_{AE}\Pi^{\mathrm{c}})\right\|_1 + \frac{1}{2}\mathds{E}_h\left\|\mathcal{N}^h(\rho_{AE}\Pi)\right\|_1. \label{eq:qdecouple1}
	\end{align}
	Using the monotone decreasing of $\|\cdot \|_1$ under any CPTP map, the first term of \eqref{eq:qdecouple1} can be bounded as
	\begin{align}
		\frac{1}{2}\mathds{E}_{h} \left\|\mathcal{N}^h(\Pi^{\mathrm{c}}\rho_{AE}\Pi^{\mathrm{c}})\right\|_1
		&\leq \frac{1}{2}\mathds{E}_{h} \left\|\mathcal{R}^h(\Pi^{\mathrm{c}}\rho_{AE}\Pi^{\mathrm{c}})\right\|_1+\frac{1}{2}\mathds{E}_{h} \left\|\mathcal{U}(\Pi^{\mathrm{c}}\rho_{AE}\Pi^{\mathrm{c}})\right\|_1\notag\\
		&\leq \mathds{E}_{h} \left\|\Pi^{\mathrm{c}}\rho_{AE}\Pi^{\mathrm{c}}\right\|_1\notag\\
		&= \tr\left[\rho_{AE}\{\mathscr{P}_{\rho_E}[\rho_{AE}]>c\mathds{1}_A\otimes\rho_E\}\right]. \label{eq:one-shot_second-order4}
	\end{align}
The second term of \eqref{eq:qdecouple1} can be bounded by
	\begin{align}
		\frac{1}{2} \mathds{E}_h\left\|\mathcal{N}^h(\Pi\rho_{AE}\Pi^{\mathrm{c}})\right\|_1 &=\frac{1}{2} \mathds{E}_h\left\|\mathcal{N}^h(\Pi^{\mathrm{c}}\rho_{AE}\Pi)\right\|_1\notag\\
		& \overset{\textnormal{(a)}}{\leq} \frac{1}{2} \mathds{E}_h  \sqrt{\Tr\left[ \left(\mathcal{N}^h(\Pi^{\mathrm{c}}\rho_{AE}\Pi)\right)^{\dagger} \mathcal{N}^h(\Pi^{\mathrm{c}}\rho_{AE}\Pi) \left(\omega_C\otimes\rho_E\right)^{-1}\right]}\notag\\
		&= \frac{\sqrt{ \left|{C}\right| }}{2} \mathds{E}_h \left\| \mathcal{N}^h\left(\Pi^{\mathrm{c}}\rho_{AE}\Pi\rho_E^{-\frac12} \right) \right\|_2 \notag\\
		& \overset{\textnormal{(b)}}{\leq} \frac{\lambda\sqrt{ \left|{C}\right| }}{2} \left\|\Pi^{\mathrm{c}}\rho_{AE}\Pi\rho_E^{-\frac12} \right\|_2 \notag\\
		&= \frac{\lambda}{2} \sqrt{\left|{C}\right| \Tr\left[ \Pi \rho_{AE} \Pi^\mathrm{c} \rho_{AE} \Pi \rho_E^{-1}\right]  } \notag\\
		&\overset{\textnormal{(c)}}{\leq} \frac{\lambda}{2} \sqrt{\left|{C}\right| \Tr\left[ \Pi \rho_{AE}^2 \Pi \rho_E^{-1}\right] },\label{eq:one-shot_second-order5}
	\end{align}
	where (a) follows from Lemma~\ref{lemm:Cauchy--Swartz}, 
	(b) is due to the $\lambda$-randomizing channel given in Definition~\ref{defn:randomizing},
	and (c) follows from the operator inequality $\Pi^\mathrm{c} \leq \mathds{1}$.

 Next, note that
	\begin{align}
		\Pi \rho_E^{-1} \Pi &\leq c \, \Pi \left(\mathscr{P}_{\rho_E}[\rho_{AE}]\right)^{-1} \Pi \notag\\
		&\leq c\left(\mathscr{P}_{\rho_E}[\rho_{AE}]\right)^{-1} \notag\\
		&\leq c \left|\textnormal{\texttt{spec}}(\rho_E)\right| (\rho_{AE})^{-1},  \label{eq:one-shot_second-order6}
	\end{align}
	where we invoke the operator anti-monotonicity of matrix inversion together with the pinching inequality (Lemma~\ref{lemm:pinching}). 
	Making use of \eqref{eq:one-shot_second-order5} and \eqref{eq:one-shot_second-order6},
	\begin{align} \label{eq:one-shot_second-order7}
		\frac{1}{2} \mathds{E}_h\left\|\mathcal{N}^h(\Pi\rho_{AE}\Pi^{\mathrm{c}})\right\|_1 \leq \frac{\lambda}{2} \sqrt{ c |{C} | \left|\textnormal{\texttt{spec}}(\rho_E)\right| }.
	\end{align}
Following similar steps, one can also bound the third term of \eqref{eq:qdecouple1} by 
	\begin{align} \label{eq:one-shot_second-order8}
		\frac{1}{2} \mathds{E}_h\left\|\mathcal{N}^h(\rho_{AE}\Pi)\right\|_1 \leq \frac{\lambda}{2} \sqrt{ c |{C} | \left|\textnormal{\texttt{spec}}(\rho_E)\right| }.
	\end{align}
	Then, the inequality \eqref{eq:pmain} follows  by combining  \eqref{eq:qdecouple1}, \eqref{eq:one-shot_second-order4}, \eqref{eq:one-shot_second-order7}, and \eqref{eq:one-shot_second-order8}.

We can now generalize inequality~\eqref{eq:pmain} to all quantum states. 
	For the general non-invertible state $\rho_{AE}$ we define the approximation
	\begin{align}
		\rho_{AE}^\epsilon:= (1-\epsilon)\rho_{AE} + \epsilon\frac{\mathds{1}_A}{|{A}|}\otimes \rho_E.\notag
	\end{align}
	Moreover, since $\mathscr{P}_{\rho_E}$ is unital completely positive and trace-preserving, 
	\begin{align}
		\mathscr{P}_{\rho_E}\left[\rho_{AE}^\epsilon\right]:= (1-\epsilon)\mathscr{P}_{\rho_E}\left[\rho_{AE}\right] + \epsilon \frac{\mathds{1}_A}{|{A}|}\otimes \rho_E.\notag
	\end{align}
	It is clear that
	\begin{align}
		\lim_{\epsilon\to 0}\frac{1}{2}\mathds{E}_{h} \left\| \mathcal{R}^h(\rho_{AE}^\epsilon)-\mathcal{U}(\rho_{AE}^\epsilon) \right\|_1 = \frac{1}{2}\mathds{E}_{h} \left\| \mathcal{R}^h(\rho_{AE})-\mathcal{U}(\rho_{AE}) \right\|_1.\notag
	\end{align}
	For the right-hand side of \eqref{eq:pmain}, we have
	\begin{align}
		\left\|\mathscr{P}_{\rho_E}\left[\rho_{AE}^\epsilon\right] - \mathscr{P}_{\rho_E}[\rho_{AE}]\right\|_1 \leq 2\epsilon\notag.
	\end{align}
	Furthermore, note that with our chosen value for $c$, $\{\mathscr{P}_{\rho_E}(\rho_{AE}) = c\mathds{1}_A\otimes \rho_E\}=0$. It is then easy to prove that for a small enough $\epsilon$, the following holds
	\begin{align}
		\left\{ \mathscr{P}_{\rho_E}[(\rho_{AE})^{\epsilon}] > c \mathds{1}_A\otimes (\rho_E)^{\epsilon} \right\} = \left\{ \mathscr{P}_{\rho_E}[\rho_{AE}] > c \mathds{1}_A\otimes \rho_E \right\}.\notag
	\end{align}
	In fact, because $\mathscr{P}_{\rho_E}[(\rho_{AE})]$ and $\mathds{1}_A\otimes \rho_E$ are commutative, they can be viewed as functions on the finite set of spectrum. Therefore, we have
	\begin{align}
		\lim_{\epsilon\to 0} \tr\left[(\rho_{AE})^{\epsilon}\left\{ \mathscr{P}_{\rho_E}[(\rho_{AE})^{\epsilon}] > c \mathds{1}_A\otimes (\rho_E)^{\epsilon} \right\}\right]\notag = \tr\left[\rho_{AE}\left\{ \mathscr{P}_{\rho_E}[\rho_{AE}] > c \mathds{1}_A\otimes \rho_E \right\}\right],\notag
	\end{align}
	which proves \eqref{eq:pmain} by approximation.
 
\end{proof}

\section{Proof of Converse Bound}\label{sec:converse}
In this section we prove the converse bound of {Theorem}~\ref{theo:one-shot_second-order_PA}.
\begin{proof}[Proof of converse in {Theorem}~\ref{theo:one-shot_second-order_PA}]
	
Consider the partial trace channel $\mathcal{T}_{A\to C}=\tr_{A_1}\otimes \id_C$. Note that $|{A}| = |{A}_1||{C}|$. Denote $\rho^U_{AE} = U_A\rho_{AE}U_A^{\dagger}$. We have,
\begin{align}
\varepsilon  &\geq\int_{\mathds{U}({A})} \left\|\mathcal{T}_{A\to C}(U_A\rho_{AE}U_A^{\dagger})-\omega_{C}\otimes\rho_E\right\|_1 \,\mathrm{d} U\notag
\\
& =\mathds{E}_U \left\|\Tr_{A_1}[\rho_{AE}^U]-\omega_{C}\otimes\rho_E\right\|_1 \notag
\\
& \overset{(a)}{\geq} \mathds{E}_U\Tr\bigg[\left(\Tr_{A_1}[\rho_{AE}^U]- \omega_{C}\otimes \rho_E\right) \frac{\Tr_{A_1}[\rho_{AE}^U]}{\Tr_{A_1}[\rho_{AE}^U]+ c\omega_C\otimes\rho_E}\bigg]\notag\\
& \overset{(b)}{\geq} \mathds{E}_U\Tr\left[\Tr_{A_1}[\rho_{AE}^U]\frac{\Tr_{A_1}[\rho_{AE}^U]}{\Tr_{A_1}[\rho_{AE}^U]+ c\omega_C\otimes\rho_E}\right]- c^{-1}\notag
\\
&=\mathds{E}_U\Tr\left[\rho_{AE}^U \frac{\mathds{1}_{A_1}\otimes\Tr_{A_1}[\rho_{AE}^U]}{\mathds{1}_{A_1}\otimes(\Tr_{A_1}[\rho_{AE}^U]+ c\omega_C\otimes\rho_E)}\right]- c^{-1}\notag\\
& \overset{(c)}{\geq}  \mathds{E}_U\e^{H_{\min}(A_1|CE)_{\rho^U}} \Tr\left[\rho_{AE}^U\frac{\rho_{AE}^U}{\mathds{1}_{A_1}\otimes(\Tr_{A_1}[\rho_{AE}^U]+ c\omega_C\otimes\rho_E)}\right] - c^{-1}\notag\\
& \overset{(d)}{\geq} \frac{1}{|\mathsf{A}_1|}  \mathds{E}_U\Tr\left[\rho_{AE}^U\frac{\rho_{AE}^U}{\mathds{1}_{A_1}\otimes(\Tr_{A_1}[\rho_{AE}^U]+ c\omega_C\otimes\rho_E)}\right] - c^{-1}\notag\\
& = \frac{1}{|\mathsf{A}_1|}\mathds{E}_U \exp D^*_2\left(\rho^U_{AE}\big{|}\big{|}\mathds{1}_{A_1}\otimes (\Tr_{A_1}[\rho^U_{AE}]+c\omega_C\otimes \rho_E)\right)- c^{-1}.\notag
\end{align}
where in (a) we use {Lemma}~\ref{lemm:1norm} with $\Pi = \frac{\Tr_{A_1}[\rho_{AE}^U]}{\Tr_{A_1}[\rho_{AE}^U]+ c\omega_C\otimes\rho_E}$, the inequality (c) is followed by the definition of $H_{\min}$, in (d) $H_{\min}(A_1|CE) \geq -\log(\text{rank}(\rho_{A_1}))\geq -\log|{A}_1|$ is used, and in (b) we apply the following inequality
\begin{align}
\mathds{E}_U\Tr\left[\omega_{C}\otimes \rho_E\frac{\Tr_{A_1}[\rho_{AE}^U]}{\Tr_{A_1}[\rho_{AE}^U]+ c\omega_C\otimes\rho_E}\right] 
& = c^{-1}\mathds{E}_{U}\Tr\left[c\omega_{C}\otimes \rho_E\frac{\Tr_{A_1}[\rho_{AE}^U]}{\Tr_{A_1}[\rho_{AE}^U]+ c\omega_C\otimes\rho_E}\right]\notag\\
& \leq c^{-1}.\notag
\end{align}

Next, let $\sigma^U_{CE} = \Tr_{A_1}[\rho^U_{AE}]+ c\omega_C\otimes \rho_E$. We change the position of unitary operator and expectation
\begin{align}
\mathds{E}_U \exp D^*_2\left(\rho^U_{AE}\Vert\mathds{1}_{A_1}\otimes \sigma^U_{CE}\right)
& = \mathds{E}_U \exp D^*_2\left(U_A\rho_{AE}U_A^\dagger\Vert\mathds{1}_{A_1}\otimes \sigma^U_{CE}\right)\notag\\
& \geq \mathds{E}_U \exp D^*_2\left(\rho_{AE}\Vert U_A^{\dagger}(\mathds{1}_{A_1}\otimes \sigma^U_{CE})U_A\right)\notag\\
& \geq  \exp D^*_2\left(\rho_{AE}\big{|}\big{|}\mathds{E}_U U_A^{\dagger}\left(\mathds{1}_{A_1}\otimes \sigma^U_{CE}\right)U_A\right).\notag
\end{align}
The first inequality is followed by \cite[Theorem~1 (Data processing)]{BEI17}, and in the second inequality we use the joint convexity of the operation stated in Lemma \ref{lemm:joint_convexity}.

Using Lemma \ref{expectUAE}, which is stated after the proof, by setting $A_2 = C$ and $\sigma_{AE} = \rho_{AE} + c\frac{\mathds{1}_A}{|{A}|}\otimes\rho_E$,
\begin{align}
\mathds{E}_U U_A^{\dagger}\left(\mathds{1}_{A_1}\otimes \sigma^U_{CE}\right)U_A 
& = \mathds{E}_U U_A^{\dagger}\left(\mathds{1}_{A_1}\otimes \left(\Tr_{A_1}\left[U_A\sigma_{AE}U_A^{\dagger}\right]\right)\right)U_A \notag\\
&= \alpha (1+c)\mathds{1}_A\otimes \rho_E + \beta\left(\rho_{AE} + c\frac{\mathds{1}_A}{|A|}\otimes \rho_E\right)\notag \\
& = \beta\rho_{AE} + \left(\alpha(1+c)+ c\frac{\beta}{|A|}\right)\mathds{1}_A\otimes \rho_E.\notag
\end{align}
where $\alpha = \frac{|{A}||{A}_1|-|{C}|}{|{A}|^2-1} = \frac{1}{|{C}|}\frac{|{A}|^2-|{C}|^2}{|{A}|^2-1}$ and $\beta = \frac{|{A}||{C}|-|{A}_1|}{|{A}|^2-1} = \frac{|{C}|}{|{A}|}\frac{|{A}|^2-\frac{|{A}|^2}{|{C}|^2}}{|{A}|^2-1}$.

Combining above equations,
\begin{align}
\varepsilon & \geq \frac{1}{|{A}_1|}  \e^{D^*_2\left(\rho_{AE}\big{|}\big{|}\beta\rho_{AE} + \left(\alpha(1+c)+ c\frac{\beta}{|{A}|}\right)\mathds{1}_A\otimes \rho_E\right)}- c^{-1}\notag\\
& \geq \frac{|{C}|}{|{A}|} (\delta+ \varepsilon)\cdot \bigg(\beta + \bigg(\alpha(1+c)+ c\frac{\beta}{|{A}|}\bigg)\e^{-D_s^{1-\varepsilon-\delta}(\rho_{AE}\Vert\mathds{1}_A\otimes \rho_E)}\bigg)^{-1} - c^{-1},\notag
\end{align}
where in the second inequality we use Lemma \ref{lemm:D2toDs} by setting $\lambda_1 = \beta$, $\lambda_2 = (\alpha(1+c) + \frac{c\beta}{|{A}|})$, and $\eta = 1-\varepsilon - \delta$. 
Rearranging the terms, 
\begin{align}
\frac{|{A}|}{|{C}|}\beta + \frac{|{A}|}{|{C}|}\left(\alpha(1+c)+ c\frac{\beta}{|{A}|}\right)\e^{-D_s^{1-\varepsilon-\delta}(\rho_{AE}\Vert\mathds{1}_A\otimes \rho_E)}
\geq \frac{\delta + \varepsilon}{\varepsilon + c^{-1}}.\notag
\end{align}
Since $\frac{|{A}|}{|{C}|}\beta \leq 1$, $c \leq 1+c$, and $\alpha + \frac{\beta}{|{A}|} = \frac{1}{|{C}|}$,
\begin{align}
\log\frac{|{A}|}{|{C}|^2} + \log(1+c) - D^{1-\varepsilon-\delta}_s(\rho_{AE}\big{|}\big{|}\mathds{1}_A\otimes\rho_E) \geq \log\left(\frac{\delta - c^{-1}}{\varepsilon+ c^{-1}}\right).\label{eq:conversepf}
\end{align}
Applying {Lemma}~\ref{lemm:relation} on \eqref{eq:conversepf} and we complete the proof.
\end{proof}

\begin{lemm}\label{expectUAE}
Let $|{A}|= |{A}_1||{A}_2|$ and $U_A \in \mathds{U}({A})$ be chosen from the normalized Haar measure. For all $\sigma_{AE} \in \mathcal{L}({A}\otimes {E})$,
\begin{align}
\mathds{E}_{U} U_A^{\dagger}\big(\mathds{1}_{A_1}\otimes \Tr_{A_1}\left[U_A\sigma_{AE} U_A^{\dagger}\right]\big)U_A = \alpha \mathds{1}_A\otimes\sigma_E + \beta \sigma_{AE},\notag
\end{align}
where $\alpha = \frac{|{A}||{A}_1|-|{A}_2|}{|{A}|^2-1}$ and $\beta = \frac{|{A}||{A}_2|-|{A}_1|}{|{A}|^2-1}$.
\end{lemm}
\begin{proof}[Proof of Lemma~\ref{expectUAE}]
Let $\tau_{A_1}(\cdot) = \Tr_{A_1}(\cdot)\mathds{1}_{A_1}$.  For a system $X$ and its copy $X'$, denote 
\[{\Gamma_{X'X} := |\mathcal{X}|\Phi_{X'X} = \sum_{i,j\in \mathcal{X}} E_{i,j}^{X'}\otimes E_{i,j}^{X}}
\]
as the non-normalized maximally entangled state. Besides, let $\mathcal{F}_{X}$ be the swap operator between the system and its duplicated space $X'$. Note that
\begin{align}
 &U_A^{\dagger}\Big(\mathds{1}_{A_1}\otimes \Tr_{A_1}\left[U_A\sigma_{AE} U_A^{\dagger}\right]\Big)U_A \notag\\
 &= U_A^{\dagger} \left(\tau_{A_1}\otimes \id_{A_2}\otimes \id_E \left(U_A \sigma_{AE} U_A^{\dagger}\right)\right) U_A \notag\\
 & = U_A^{\dagger} \Tr_{A'_1A'_2E'}\bigg[\left(\mathds{1}_{A'_1A_1}\otimes \Gamma_{A'_2A_2}\otimes \Gamma_{E'E}\right)^{T_{A'_1A'_2E'}}\left(U_{A'}\sigma_{A'E'}U_{A'}^{\dagger}\otimes \mathds{1}_{A_1A_2E}\right)\bigg]U_A \notag\\
 & = \Tr_{A'E'}\bigg[U_{A'}^{\dagger}\otimes U_{A}^{\dagger} \left(\mathds{1}_{A'_1A_1}\otimes \Gamma_{A'_2A_2}\otimes \Gamma_{E'E}\right)^{T_{A'_1A'_2E'}} U_{A'}\otimes U_{A} \left(\sigma_{A'E'}\otimes \mathds{1}_{AE}\right)\bigg]\notag\\
 & = \Tr_{A'E'}\big[U_{A'}^{\dagger}\otimes U_{A}^{\dagger} \left(\mathds{1}_{A'_1A_1}\otimes \mathcal{F}_{A_2} \otimes \mathcal{F}_{E}\right)  U_{A'}\otimes U_{A} \left(\sigma_{A'E'}\otimes \mathds{1}_{AE}\right)\big].\label{transform2space}
\end{align}
The second equation follows from {Lemma}~\ref{Choi} and the fact that the Choi operator of $\tau_{A_1}\otimes \id_{A_2}\otimes \id_E$ is $\mathds{1}_{A'_1A_1}\otimes \Gamma_{A'_2A_2}\otimes \Gamma_{E'E}$.

Using Lemma 3.4 in \cite{FRE10},
\begin{align}
\mathds{E}_{U} U_{A'}^{\dagger}\otimes U_{A}^{\dagger} \left(\mathds{1}_{A'_1A_1}\otimes \mathcal{F}_{A_2} \right)U_{A'}\otimes U_{A} = \alpha \mathds{1}_{A'}\otimes \mathds{1}_{A} + \beta \mathcal{F}_A,\label{2spaceexpect}
\end{align}
where
\begin{align}
\begin{cases}
    \alpha |{A}|^2 + \beta |{A}| = \Tr[\mathds{1}_{A'_1A_1}\otimes \mathcal{F}_{A_2}]=|{A}_1|^2 |{A}_2| \\
    \alpha|{A}| + \beta |{A}|^2 = \Tr[(\mathds{1}_{A'_1A_1}\otimes \mathcal{F}_{A_2})\mathcal{F}_A] = |{A}_1| |{A}_2|^2,\notag
\end{cases}
\end{align}
or equivalently $\alpha = \frac{|{A}||{A}_1|-|{A}_2|}{|{A}|^2-1}$ and $\beta = \frac{|{A}||{A}_2|-|{A}_1|}{|{A}|^2-1}$.

Thus, using \eqref{transform2space} and \eqref{2spaceexpect}, 
\begin{align}
&\mathds{E}_{U} \Big(U_A^{\dagger}\mathds{1}_{A_1}\otimes \Tr_{A_1}\left[U_A\sigma_{AE} U_A^{\dagger}\right]U_A \Big)\notag\\
& = \mathds{E}_{U} \Tr_{A'E'}\big[U_{A'}^{\dagger}\otimes U_{A}^{\dagger} \left(\mathds{1}_{A'_1A_1}\otimes \mathcal{F}_{A_2} \otimes \mathcal{F}_{E}\right) U_{A'}\otimes U_{A} \left(\sigma_{A'E'}\otimes \mathds{1}_{AE}\right)\big]\notag\\
& = \Tr_{A'E'}\bigg[ \left(\mathds{E}_{U}U_{A'}^{\dagger}\otimes U_{A}^{\dagger} \left(\mathds{1}_{A'_1A_1}\otimes \mathcal{F}_{A_2} \right)U_{A'}\otimes U_{A}\right) \otimes \mathcal{F}_E\left(\sigma_{A'E'}\otimes \mathds{1}_{AE}\right)\bigg]\notag\\
& = \Tr_{A'E'}\left[ (\alpha \mathds{1}_{A'}\otimes \mathds{1}_{A} + \beta \mathcal{F}_A)\otimes \mathcal{F}_E\left(\sigma_{A'E'}\otimes \mathds{1}_{AE}\right)\right]\notag\\
& = \alpha \mathds{1}_A\otimes \sigma_E + \beta \sigma_{AE}\notag
\end{align}
where in the last equality we use Lemma~\ref{lemm:swap}, and the proof is done.
\end{proof}

\section{Application: Entanglement Distillation}\label{sec:distill}
Our result of quantum decoupling can be applied to entanglement distillation. As mentioned in the introduction, we consider the case when Alice and Bob only apply local operations and one-way classical communications $\Pi_{1-\text{LOCC}}$.
Using the achievability bound of Theorem~\ref{theo:one-shot_second-order_PA}, we have the following lower bound of 
maximal distillable entanglement $E_{\textnormal{ED}}^{\varepsilon}$.
\begin{theo}\label{theo:EA}
For any bipartite state $\rho_{AB}$, $\varepsilon \in (0,1)$, and $\delta \in (0, \frac{\varepsilon}{3})$, there exists an entanglement distillation protocol such that, 
\begin{align}
			\log E_{\textnormal{ED}}^{\varepsilon}(\rho) \geq H_\textnormal{h}^{1-\varepsilon + 3\delta}(A{\,|\,}B)_{\rho}  -\log\frac{\nu^2}{\delta^4}\notag
\end{align}
where $\nu =  |\textnormal{\texttt{spec}}(\rho_E)|$, ${\rho}_{ABE}$ is the purification of $\rho_{AB}$. Also, 
\begin{align}
\log E_{\textnormal{ED}}^{\varepsilon}(\rho^{\otimes n}) \geq nH(A|B)_{\rho} + \sqrt{nV(A|B)_{\rho}}\Phi^{-1}(\varepsilon)  + O(\log n).\notag
\end{align}
\end{theo}

\begin{remark}
For maximal distillable entanglement, \cite[Proposition~21]{WTB17} gives the following one-shot achievability bound when 1-LOCC operation is allowed during the task:
\begin{align}
\log E_{\textnormal{ED}}^{\varepsilon}(\rho) \gtrsim -H_{\max}^{\sqrt{2\varepsilon- \varepsilon^2}}(A|B)_{\rho}.\notag
\end{align}
Here $H_{\max}^{\varepsilon}(A|B)_{\rho}$ is the smooth conditional max-entropy of $A$ given $B$. \cite[Theorem~7]{FWTD19} expanded the above one-shot bound to second order expansion as follows:
\begin{align}
\log E_{\textnormal{ED}}^{\varepsilon}(\rho^{\otimes n})\geq nH(A|B)_{\rho} + \sqrt{nV(A|B)_{\rho}}\Phi^{-1}(2\varepsilon - \varepsilon^2)+ O(\log n)\notag,
\end{align} 
where ${\rho}_{ABE}$ is the purification of $\rho_{AB}$. Such bound is tighter than ours in second order rate. The main reason of our bound not being optimal comes from that we bound the purified distance $F_{\textnormal{ED}}$ by trace distance estimate from our quantum decoupling result. The transfer between these two different distances subsume the tightness of our bound. Nonetheless, via our result of quantum decoupling, we still obtain an achievability bound without smoothing that leads to the same first-order rate in their result. 
\end{remark}

\begin{proof}[Proof of Theorem~\ref{theo:EA}]
Let $\rho_{ABE}$ be a purification of $\rho_{AB}$. Define
\begin{align}
\mathcal{T}_{A\to A_1 X}(\cdot) := \sum_{x\in \mathcal{X}} \mathcal{V}^x_{A\to A_1}(P^x_A(\cdot)P^x_A)\otimes \ket{x}\bra{x}_{X},
\end{align}
where $\{P^x_A\}$ is a set of projectors such that $\sum_x P^x_A = \mathds{1}_A$ and for each $x$, $\mathcal{V}^x_{A\to A_1}$ is an isometry
that embeds the subspace of $P^x_A$ onto $A_1$.
It can be shown that the map $\mathcal{T}_{A\to A_1 X}(U_A(\cdot)U_A^{\dagger})$ is $\frac{1}{\sqrt{|{X}|}}$ randomizing by applying part of the steps in the proof of \cite[Theorem 3.7]{dupuis2010decoupling}. By equation \eqref{eq:pmain} (which is the core of achievability part of Theorem~\ref{theo:one-shot_second-order_PA}), for any $c$ in the set $  \{c > 0 | \{\mathscr{P}_{\rho_E}(\rho_{AE}) = c\mathds{1}_A\otimes \rho_E\}=0 \}$, there exists an unitary operator $U_A$ such that
\begin{align}
&\frac{1}{2} \left\| \mathcal{T}_{A\to A_1 X_A}(U_A\rho_{AE}U_A^{\dagger})-\frac{\mathds{1}_X}{|\mathcal{X}|}\otimes \frac{\mathds{1}_{A_1}}{|{A}_1|}\otimes \rho_{E} \right\|_1 \\
&\leq \tr\left[\rho_{AE}\left\{ \mathscr{P}_{\rho_E}[\rho_{AE}] > c \mathds{1}_A\otimes \rho_E \right\}\right] + \frac{1}{\sqrt{|{X}|}}\sqrt{ c|\textnormal{\texttt{spec}}(\rho_E)|{X}||{A}_1| }.\label{eq:distill_1}
\end{align}
Define
$\sigma^{U}_{A_1 BE X} := \mathcal{T}_{A\to A_1 X}(U_A\rho_{ABE}U_A^{\dagger}) 
= \sum_{x\in \mathcal{X}}p_X(x)\ket{x}\bra{x}_{X} \otimes \sigma^x_{A_1 B E}$.
The fidelity has the following relation:
\begin{align}
 F\left(\sigma^{U}_{A_1 E X_A}, \frac{\mathds{1}_X}{|{X}|}\otimes \frac{\mathds{1}_{A_1}}{|{A}_1|}\otimes \rho_{E}\right)
 & \overset{(a)}{=} \sum_{x\in {X}}\sqrt{p_X(x)\frac{1}{|{X}|}} F(\sigma^x_{A_1 E}, \frac{\mathds{1}_{A_1}}{|{A}_1|}\otimes \rho_{E})\notag\\
 & \overset{(b)}{=} \sum_{x\in {X}}\sqrt{p_X(x)\frac{1}{|{X}|}} F(\mathcal{K}^x_{B\to B_1\Bar{B}}(\sigma^x_{A_1 B E}), \Phi_{A_1 B_1}\otimes \tau_{\Bar{B}E})\notag\\
 & \overset{(c)}{=} F\bigg(\sum_x p_X(x)\ket{x}\bra{x}_{X}\otimes \mathcal{K}^x_{B\to B_1\Bar{B}}(\sigma^x_{A_1 B E}), \frac{\mathds{1}_X}{|{X}|}\otimes \Phi_{A_1 B_1}\otimes \tau_{\Bar{B}E}\bigg)\notag\\
 & \overset{(d)}{\leq} F\left(\Lambda_{B X \to B_1}\left(\sigma^{U}_{A_1 B X_A}\right), \Phi_{A_1 B_1}\right),\label{eq:distill_2}
\end{align}
where $\tau_{\Bar{B}E}$ is a purification of $\rho_E$, and $\Lambda_{B X\to B_1}(\cdot) = \Tr_{\Bar{B}X}\left[\sum_{x}\ket{x}\bra{x}_{X}\otimes \mathcal{K}^x_{B\to B_1 \Bar{B}}(\cdot)\right]$.
In (a) and (c) we use the fact that for c-q state $\omega_{XB} = \sum_x p(x)\ket{x}\bra{x}_{X}\otimes \omega_x$ and $\tau_{XB} = \sum_x q(x)\ket{x}\bra{x}_{X}\otimes \tau_x$, $F(\omega_{XB}, \tau_{XB}) = \sum_x \sqrt{p(x)q(x)}F(\omega_x,\tau_x)$. Equation $(b)$ comes from Uhlmann's theorem: there exists an isometry map $\mathcal{K}^x_{B\to B_1 \Bar{B}}$ such that equality $(b)$ is observed. The inequality (d) comes from the monotonicity of fidelity after partial trace.
Moreover, 
\begin{align}
F\left(\sigma^{U}_{A_1 E X_A}, \frac{\mathds{1}_X}{|{X}|}\otimes \frac{\mathds{1}_{A_1}}{|{A}_1|}\otimes \rho_{E}\right) \geq 1- \frac{1}{2}\left\| \mathcal{T}_{A\to A_1 X}(U_A\rho_{AE}U_A^{\dagger})-\frac{\mathds{1}_X}{|{X}|}\otimes \frac{\mathds{1}_{A_1}}{|{A}_1|}\otimes \rho_{E} \right\|_1.\label{eq:distill_3}
\end{align}
Combining \eqref{eq:distill_1}, \eqref{eq:distill_2}, and \eqref{eq:distill_3}, we have 
\begin{align}
&F\left(\Lambda_{B X \to B_1}\left(\mathcal{T}_{A\to A_1 X_A}(U_A\rho_{ABE}U_A^{\dagger})\right), \Phi_{A_1 B_1}\right) \\
&\geq 1- \frac{1}{2} \left\| \mathcal{T}_{A\to A_1 X_A}(U_A\rho_{AE}U_A^{\dagger})-\frac{\mathds{1}_X}{|{X}|}\otimes \frac{\mathds{1}_{A_1}}{|{A}_1|}\otimes \rho_{E} \right\|_1\notag\\
& \geq 1- \tr\left[\rho_{AE}\left\{ \mathscr{P}_{\rho_E}[\rho_{AE}] > c \mathds{1}_A\otimes \rho_E \right\}\right] + \sqrt{ c|\textnormal{\texttt{spec}}(\rho_E)|{A}_1| }.\notag
\end{align}

Let $c = \exp\left\{D_\text{s}^{1-\varepsilon+ \delta}(\mathscr{P}_{\rho_E}[\rho_{AE}]  \,\|\,\mathds{1}_A\otimes\rho_E) + \xi\right\}$ for some small $\xi >0$ and $0<\delta<\eps$. 
Following similar steps in the proof of decoupling achievability bound in Section~\ref{sec:lower}, 
we have 
\begin{align}
\log |{A}_1| 
\geq H_\text{h}^{1-\varepsilon + 3\delta}(A{\,|\,}E)_{\rho} - \xi -\log \frac{|\texttt{spec}(\rho_E)|^2}{\delta^4}. \label{eq:distillbd}
\end{align}

We can now apply the 1-LOCC operation $\Pi_{1-\text{LOCC}}(\cdot) = \Lambda(\mathcal{T}(U_A(\cdot)U_A^{\dagger}))$ by the following steps:
For Alice and Bob holding $\rho_{AB}$, Alice applies the unitary map $U_A(\cdot)$ followed by the map $T_{A\to A_1 X_A}$. She then transmit $x\in X_A$ to Bob classically so Bob receive $x\in X_B$. Bob then applies the decoding map $\Lambda_{BX_B\to B_1}$ to generate the shared maximally entangled state.
Since the maximally entangled state extracted is of dimension $|{A}_1|$, the proof is done by using \eqref{eq:distillbd} and the definition of maximal distillable entanglement $E_{\textnormal{ED}}^{\varepsilon}(\rho)$. The bound of $E_{\textnormal{ED}}^{\varepsilon}(\rho^{\otimes n})$ follows from the above one-shot bound and Lemma~\ref{lemm:second}.
\end{proof}

\section{Conclusions} \label{sec:conclusion}
 We establish a tight one-shot characterization of maximal remainder dimension for quantum decoupling, with the trace distance as the security criterion. Our characterization, utilizing the conditional hypothesis testing entropy, is obtained via direct analysis of trace distance without using smoothed entropy and relating to the purified distance. Most importantly, our one-shot bound leads to exact second-order rate in i.i.d.~asymptotic scenario, a problem that has been open for a decade. 
 Moreover, our result has an application in entanglement distillation. Applying the lower bound of $\ell^{\varepsilon}(A{\,|\,}E)_{\rho}$, we found an achievability bound of the size of the maximally entangled state Alice and Bob can produce via $1$-LOCC operation. 

 We shall remark that most entanglement distillation achievability proofs via quantum decoupling are through Uhlmann's theorem.
 A straightforward application of the Fuchs---van de Graaf for translating trace distance to fidelity will incur a square-root loose on the error bound.
 It is not clear yet such translation can be sharpened or there is any way to bypass Uhlmann's theorem when applying quantum decoupling.
 If it is the case, then the established one-shot decoupling bound on $\ell^{\varepsilon}(A{\,|\,}E)_{\rho}$ could lead to tighter characterizations for entanglement distillation.
We leave this interesting question for future work. 

\section*{Acknowledgement}
H.-C.~Cheng would like to thank Kai-Min Chung for his insightful discussions, and also thank Marco Tomamichel and Mario Berta for helpful comments and pointing out relevant references.
H.-C.~Cheng is supported by the Young Scholar Fellowship (Einstein Program) of the National Science and Technology Council, Taiwan (R.O.C.) under Grants No.~NSTC 112-2636-E-002-009, No.~NSTC 113-2119-M-007-006, No.~NSTC 113-2119-M-001-006, No.~NSTC 113-2124-M-002-003, by the Yushan Young Scholar Program of the Ministry of Education, Taiwan (R.O.C.) under Grants No.~NTU-112V1904-4 and by the research project ``Pioneering Research in Forefront Quantum Computing, Learning and Engineering'' of National Taiwan University under Grant No. NTU-CC- 112L893405 and NTU-CC-113L891605. H.-C.~Cheng acknowledges the support from the “Center for Advanced Computing and Imaging in Biomedicine (NTU-113L900702)” through The Featured Areas Research Center Program within the framework of the Higher Education Sprout Project by the Ministry of Education (MOE) in Taiwan.
LG's research was partially supported by NSF Grant DMS-2154903.

{\larger
\bibliographystyle{myIEEEtran}
\bibliography{reference.bib}

\begin{thebibliography}{10}
\providecommand{\url}[1]{#1}
\csname url@samestyle\endcsname
\providecommand{\newblock}{\relax}
\providecommand{\bibinfo}[2]{#2}
\providecommand{\BIBentrySTDinterwordspacing}{\spaceskip=0pt\relax}
\providecommand{\BIBentryALTinterwordstretchfactor}{4}
\providecommand{\BIBentryALTinterwordspacing}{\spaceskip=\fontdimen2\font plus
\BIBentryALTinterwordstretchfactor\fontdimen3\font minus
  \fontdimen4\font\relax}
\providecommand{\BIBforeignlanguage}[2]{{%
\expandafter\ifx\csname l@#1\endcsname\relax
\typeout{** WARNING: IEEEtran.bst: No hyphenation pattern has been}%
\typeout{** loaded for the language `#1'. Using the pattern for}%
\typeout{** the default language instead.}%
\else
\language=\csname l@#1\endcsname
\fi
#2}}
\providecommand{\BIBdecl}{\relax}
\BIBdecl

\bibitem{HOW07}
M.~Horodecki, J.~Oppenheim, and A.~Winter, ``Quantum state merging and negative
  information,'' \emph{Comm. Math. Phys., (269, 107), quant-ph/0512247}, 2007.

\bibitem{BCR11}
M.~Berta, M.~Christandl, and R.~Renner, ``The quantum reverse {Shannon} theorem
  based on one-shot information theory,''
  \href{http://dx.doi.org/10.1007/s00220-011-1309-7}{\emph{Communications in
  Mathematical Physics}},
  \href{http://dx.doi.org/10.1007/s00220-011-1309-7}{vol. 306},
  \href{http://dx.doi.org/10.1007/s00220-011-1309-7}{no.~3},
  \href{http://dx.doi.org/10.1007/s00220-011-1309-7}{pp. 579--615},
  \href{http://dx.doi.org/10.1007/s00220-011-1309-7}{aug 2011}.

\bibitem{BDH+09}
C.~H. Bennett, I.~Devetak, A.~W. Harrow, P.~W. Shor, and A.~Winter, ``Quantum
  reverse shannon theorem,'' \emph{arXiv:0912.5537v2}, 2009.

\bibitem{Bus09}
F.~Buscemi, ``Private quantum decoupling and secure disposal of information,''
  \emph{New Journal of Physics, 11:123002}, 2009.

\bibitem{GPW05}
B.~Groisman, S.~Popescu, and A.~Winter, ``Quantum, classical, and total amount
  of correlations in quantum state,'' \emph{Physical Review A, 72:032317},
  2005.

\bibitem{LPSW09}
N.~Linden, S.~Popescu, A.~J. Short, and A.~Winter, ``Quantum mechanical
  evolution towards thermal equilibrium,'' \emph{Physical Review E, 79:061103},
  2009.

\bibitem{HP07}
P.~Hayden and J.~Preskill, ``Black holes as mirrors: quantum information in
  random subsystems,'' \emph{Journal of High Energy Physics, 07:120}, 2007.

\bibitem{BH13}
F.~G. Brandao and M.~Horodecki, ``An area law for entanglement from exponential
  decay of correlations,'' \emph{Nature Physics, advance online publication},
  2013.

\bibitem{Tom16}
M.~Tomamichel, \emph{Quantum Information Processing with Finite
  Resources}.\hskip 1em plus 0.5em minus 0.4em\relax Springer International
  Publishing, 2016.

\bibitem{TCR10}
M.~Tomamichel, R.~Colbeck, and R.~Renner, ``Duality between smooth min- and
  max-entropies,''
  \href{http://dx.doi.org/10.1109/tit.2010.2054130}{\emph{{IEEE} Transactions
  on Information Theory}},
  \href{http://dx.doi.org/10.1109/tit.2010.2054130}{vol.~56},
  \href{http://dx.doi.org/10.1109/tit.2010.2054130}{no.~9},
  \href{http://dx.doi.org/10.1109/tit.2010.2054130}{pp. 4674--4681},
  \href{http://dx.doi.org/10.1109/tit.2010.2054130}{sep 2010}.

\bibitem{TH13}
M.~Tomamichel and M.~Hayashi, ``A {Hierarchy} of {Information} {Quantities} for
  {Finite} {Block} {Length} {Analysis} of {Quantum} {Tasks},''
  \href{http://dx.doi.org/10.1109/TIT.2013.2276628}{\emph{IEEE Transactions on
  Information Theory}},
  \href{http://dx.doi.org/10.1109/TIT.2013.2276628}{vol.~59},
  \href{http://dx.doi.org/10.1109/TIT.2013.2276628}{no.~11},
  \href{http://dx.doi.org/10.1109/TIT.2013.2276628}{pp. 7693--7710},
  \href{http://dx.doi.org/10.1109/TIT.2013.2276628}{Nov. 2013},
  arXiv:1208.1478.

\bibitem{MBD+17}
C.~Majenz, M.~Berta, F.~Dupuis, R.~Renner, and M.~Christandl, ``Catalytic
  decoupling of quantum information,''
  \href{http://dx.doi.org/10.1103/PhysRevLett.118.080503}{\emph{Phys. Rev.
  Lett.}}, \href{http://dx.doi.org/10.1103/PhysRevLett.118.080503}{vol. 118},
  \href{http://dx.doi.org/10.1103/PhysRevLett.118.080503}{p. 080503},
  \href{http://dx.doi.org/10.1103/PhysRevLett.118.080503}{Feb 2017}.

\bibitem{KL21}
K.~Li, Y.~Yao, and M.~Hayashi, ``Tight exponential analysis for smoothing the
  max-relative entropy and for quantum privacy amplification,''
  \emph{arXiv:2111.01075 [quant-ph]}, 2022.

\bibitem{LY21a}
K.~Li and Y.~Yao, ``Reliability function of quantum information decoupling,''
  \emph{arXiv:2111.06343 [quant-ph]}, 2021.

\bibitem{Dup14}
F.~Dupuis, M.~Berta, J.~Wullschleger, and R.~Renner, ``One-shot decoupling,''
  \emph{arXiv:1012.6044v3 [quant-ph]}, 2014.

\bibitem{R99}
E.~M. Rains, ``Bound on distillable entanglement,'' \emph{Phys. Rev. A, vol.
  60, no. 1, pp. 179–184, Jul.}, 1999.

\bibitem{R01}
E.~Rains, ``A semidefinite program for distillable entanglement,'' \emph{IEEE
  Trans. Inform. Theory, vol. 47, no. 7, pp. 2921–2933}, 2001.

\bibitem{HHH00}
M.~Horodecki, P.~Horodecki, and R.~Horodecki, ``Asymptotic manipulations of
  entanglement can exhibit genuine irreversibility,'' \emph{Phys. Rev. Lett.,
  vol. 84, no. 19, pp. 4260–4263}, 2000.

\bibitem{CW04}
M.~Christandl and A.~Winter, ``“squashed entanglement”: An additive
  entanglement measure,'' \emph{Journal of Mathematical Physics, vol. 45, no.
  3, pp. 829–840}, 2004.

\bibitem{LDS18}
F.~Leditzky, N.~Datta, and G.~Smith, ``Useful states and entanglement
  distillation,'' \emph{IEEE Trans. Inform. Theory, vol. 64, no. 7, pp.
  4689–4708}, 2018.

\bibitem{WD16}
X.~Wang and R.~Duan, ``Improved semidefinite programming upper bound on
  distillable entanglement,'' \emph{Phys. Rev. A, vol. 94, no. 5, p. 050301},
  2016.

\bibitem{FWTD19}
K.~Fang, X.~Wang, M.~Tomamichel, and R.~Duan, ``Non-asymptotic entanglement
  distillation,'' \emph{arXiv:1706.0622}, 2019.

\bibitem{DL15}
N.~Datta and F.~Leditzky, ``Second-order asymptotics for source coding, dense
  coding, and pure-state entanglement conversions,'' \emph{IEEE Trans. Inform.
  Theory, vol. 61, no. 1, pp. 582–608}, 2015.

\bibitem{CHR18}
C.~Majenz, ``Entropy in quantum information theory -- communication and
  cryptography,'' \emph{arXiv:1810.10436 [quant-ph]}, 2018.

\bibitem{HN03}
M.~Hayashi and H.~Nagaoka, ``General formulas for capacity of classical-quantum
  channels,'' \href{http://dx.doi.org/10.1109/tit.2003.813556}{\emph{{IEEE}
  Transaction on Information Theory}},
  \href{http://dx.doi.org/10.1109/tit.2003.813556}{vol.~49},
  \href{http://dx.doi.org/10.1109/tit.2003.813556}{no.~7},
  \href{http://dx.doi.org/10.1109/tit.2003.813556}{pp. 1753--1768},
  \href{http://dx.doi.org/10.1109/tit.2003.813556}{Jul 2003}.

\bibitem{NH07}
H.~Nagaoka and M.~Hayashi, ``An information-spectrum approach to classical and
  quantum hypothesis testing for simple hypotheses,''
  \href{http://dx.doi.org/10.1109/tit.2006.889463}{\emph{{IEEE} Transactions on
  Information Theory}},
  \href{http://dx.doi.org/10.1109/tit.2006.889463}{vol.~53},
  \href{http://dx.doi.org/10.1109/tit.2006.889463}{no.~2},
  \href{http://dx.doi.org/10.1109/tit.2006.889463}{pp. 534--549},
  \href{http://dx.doi.org/10.1109/tit.2006.889463}{Feb 2007}.

\bibitem{WR13}
L.~{Wang} and R.~{Renner}, ``{One-Shot Classical-Quantum Capacity and
  Hypothesis Testing},''
  \href{http://dx.doi.org/10.1103/PhysRevLett.108.200501}{\emph{Physical Review
  Letters}}, \href{http://dx.doi.org/10.1103/PhysRevLett.108.200501}{vol. 108},
  \href{http://dx.doi.org/10.1103/PhysRevLett.108.200501}{no.~20},
  \href{http://dx.doi.org/10.1103/PhysRevLett.108.200501}{p. 200501},
  \href{http://dx.doi.org/10.1103/PhysRevLett.108.200501}{May 2012}.

\bibitem{Li14}
K.~Li, ``Second-order asymptotics for quantum hypothesis testing,''
  \href{http://dx.doi.org/10.1214/13-aos1185}{\emph{The Annals of Statistics}},
  \href{http://dx.doi.org/10.1214/13-aos1185}{vol.~42},
  \href{http://dx.doi.org/10.1214/13-aos1185}{no.~1},
  \href{http://dx.doi.org/10.1214/13-aos1185}{pp. 171--189},
  \href{http://dx.doi.org/10.1214/13-aos1185}{Feb 2014}.

\bibitem{Ume62}
H.~Umegaki, ``Conditional expectation in an operator algebra. {IV}. entropy and
  information,'' \href{http://dx.doi.org/10.2996/kmj/1138844604}{\emph{Kodai
  Mathematical Seminar Reports}},
  \href{http://dx.doi.org/10.2996/kmj/1138844604}{vol.~14},
  \href{http://dx.doi.org/10.2996/kmj/1138844604}{no.~2},
  \href{http://dx.doi.org/10.2996/kmj/1138844604}{pp. 59--85},
  \href{http://dx.doi.org/10.2996/kmj/1138844604}{1962}.

\bibitem{HP91}
F.~Hiai and D.~Petz, ``The proper formula for relative entropy and its
  asymptotics in quantum probability,''
  \href{http://dx.doi.org/10.1007/bf02100287}{\emph{Communications in
  Mathematical Physics}}, \href{http://dx.doi.org/10.1007/bf02100287}{vol.
  143}, \href{http://dx.doi.org/10.1007/bf02100287}{no.~1},
  \href{http://dx.doi.org/10.1007/bf02100287}{pp. 99--114},
  \href{http://dx.doi.org/10.1007/bf02100287}{Dec 1991}.

\bibitem{Wat18}
J.~Watrous, \emph{The Theory of Quantum Information}.\hskip 1em plus 0.5em
  minus 0.4em\relax Cambridge University Press, apr 2018.

\bibitem{Ren05}
R.~Renner, ``Security of quantum key distribution,''
  \href{http://dx.doi.org/10.3929/ethz-a-005115027}{\emph{Ph.D.~Thesis (ETH),
  arXiv:quant-ph/0512258}},
  \href{http://dx.doi.org/10.3929/ethz-a-005115027}{2005}.

\bibitem{Hay17}
M.~Hayashi, \emph{Quantum Information Theory}.\hskip 1em plus 0.5em minus
  0.4em\relax Springer Berlin Heidelberg, 2017.

\bibitem{Hay02}
------, ``Optimal sequence of quantum measurements in the sense of {Stein's}
  lemma in quantum hypothesis testing,''
  \href{http://dx.doi.org/10.1088/0305-4470/35/50/307}{\emph{Journal of Physics
  A: Mathematical and General}},
  \href{http://dx.doi.org/10.1088/0305-4470/35/50/307}{vol.~35},
  \href{http://dx.doi.org/10.1088/0305-4470/35/50/307}{no.~50},
  \href{http://dx.doi.org/10.1088/0305-4470/35/50/307}{pp. 10\,759--10\,773},
  \href{http://dx.doi.org/10.1088/0305-4470/35/50/307}{Dec 2002}.

\bibitem{CTT2017}
C.~T. Chubb, V.~Y.~F. Tan, and M.~Tomamichel, ``Moderate deviation analysis for
  classical communication over quantum channels,''
  \href{http://dx.doi.org/10.1007/s00220-017-2971-1}{\emph{Communications in
  Mathematical Physics}},
  \href{http://dx.doi.org/10.1007/s00220-017-2971-1}{vol. 355},
  \href{http://dx.doi.org/10.1007/s00220-017-2971-1}{no.~3},
  \href{http://dx.doi.org/10.1007/s00220-017-2971-1}{pp. 1283--1315},
  \href{http://dx.doi.org/10.1007/s00220-017-2971-1}{Nov 2017}.

\bibitem{FL13}
R.~L. Frank and E.~H. Lieb, ``Monotonicity of a relative {R{\'e}nyi} entropy,''
  \href{http://dx.doi.org/10.1063/1.4838835}{\emph{Journal of Mathematical
  Physics}}, \href{http://dx.doi.org/10.1063/1.4838835}{vol.~54},
  \href{http://dx.doi.org/10.1063/1.4838835}{no.~12},
  \href{http://dx.doi.org/10.1063/1.4838835}{p. 122201},
  \href{http://dx.doi.org/10.1063/1.4838835}{2013}.

\bibitem{MDS+13}
M.~M{\"u}ller-Lennert, F.~Dupuis, O.~Szehr, S.~Fehr, and M.~Tomamichel, ``On
  quantum {R{\'e}nyi} entropies: A new generalization and some properties,''
  \href{http://dx.doi.org/10.1063/1.4838856}{\emph{Journal of Mathematical
  Physics}}, \href{http://dx.doi.org/10.1063/1.4838856}{vol.~54},
  \href{http://dx.doi.org/10.1063/1.4838856}{no.~12},
  \href{http://dx.doi.org/10.1063/1.4838856}{p. 122203},
  \href{http://dx.doi.org/10.1063/1.4838856}{2013}.

\bibitem{WWY14}
M.~M. Wilde, A.~Winter, and D.~Yang, ``Strong converse for the classical
  capacity of entanglement-breaking and {Hadamard} channels via a sandwiched
  {R{\'{e}}nyi} relative entropy,''
  \href{http://dx.doi.org/10.1007/s00220-014-2122-x}{\emph{Communications in
  Mathematical Physics}},
  \href{http://dx.doi.org/10.1007/s00220-014-2122-x}{vol. 331},
  \href{http://dx.doi.org/10.1007/s00220-014-2122-x}{no.~2},
  \href{http://dx.doi.org/10.1007/s00220-014-2122-x}{pp. 593--622},
  \href{http://dx.doi.org/10.1007/s00220-014-2122-x}{Jul 2014}.

\bibitem{Hel67}
C.~W. Helstrom, ``Detection theory and quantum mechanics,''
  \href{http://dx.doi.org/10.1016/s0019-9958(67)90302-6}{\emph{Information and
  Control}}, \href{http://dx.doi.org/10.1016/s0019-9958(67)90302-6}{vol.~10},
  \href{http://dx.doi.org/10.1016/s0019-9958(67)90302-6}{no.~3},
  \href{http://dx.doi.org/10.1016/s0019-9958(67)90302-6}{pp. 254--291},
  \href{http://dx.doi.org/10.1016/s0019-9958(67)90302-6}{mar 1967}.

\bibitem{Hol72}
A.~Holevo, ``The analogue of statistical decision theory in the noncommutative
  probability theory,'' \emph{Proc. Moscow Math. Soc.}, vol.~26, pp. 133--149,
  1972.

\bibitem{NC09}
M.~A. Nielsen and I.~L. Chuang, \emph{Quantum Computation and Quantum
  Information}.\hskip 1em plus 0.5em minus 0.4em\relax Cambridge University
  Press, 2009.

\bibitem{BEI17}
S.~Beigi and A.~Gohari, ``Quantum achievability proof via collision relative
  entropy,'' \emph{IEEE Transactions on Information Theory 60(12), 7980-7986},
  2014.

\bibitem{Mar20}
S.~Khatri and M.~M.Wilde, ``Principles of quantum communication theory: A
  modern approach,'' \emph{arXiv:2011.04672}, 2020.

\bibitem{CH17}
H.-C. Cheng and M.-H. Hsieh, ``Moderate deviation analysis for
  classical-quantum channels and quantum hypothesis testing,''
  \href{http://dx.doi.org/10.1109/TIT.2017.2781254}{\emph{{IEEE} Transactions
  on Information Theory}},
  \href{http://dx.doi.org/10.1109/TIT.2017.2781254}{vol.~64},
  \href{http://dx.doi.org/10.1109/TIT.2017.2781254}{no.~2},
  \href{http://dx.doi.org/10.1109/TIT.2017.2781254}{pp. 1385--1403},
  \href{http://dx.doi.org/10.1109/TIT.2017.2781254}{feb 2018}.

\bibitem{Dup21}
F.~Dupuis, \href{http://dx.doi.org/10.1109/tit.2023.3301812}{\emph{IEEE
  Transactions on Information Theory}},
  \href{http://dx.doi.org/10.1109/tit.2023.3301812}{vol.~69},
  \href{http://dx.doi.org/10.1109/tit.2023.3301812}{no.~12},
  \href{http://dx.doi.org/10.1109/tit.2023.3301812}{pp. 7784--7792},
  \href{http://dx.doi.org/10.1109/tit.2023.3301812}{2023}, arXiv:2105.05342
  [quant-ph].

\bibitem{dupuis2010decoupling}
------, ``The decoupling approach to quantum information theory,'' \emph{arXiv
  preprint arXiv:1004.1641}, 2010.

\bibitem{FRE10}
------, ``The decoupling approach to quantum information theory,''
  \emph{arXiv:1004.1641 [quant-ph]}, 2010.

\bibitem{WTB17}
M.~M. Wilde, M.~Tomamichel, and M.~Berta, ``Converse bounds for private
  communication over quantum channels,'' \emph{arXiv:1602.08898}, 2017.

\end{thebibliography}
}

\end{document}